\documentclass[11pt,a4paper]{article} 
\usepackage{jheppub}
\usepackage[T1]{fontenc}
\usepackage[utf8x]{inputenc}
\usepackage{tikz-cd}

\usepackage{amssymb}
\usepackage{amsfonts}
\usepackage{amsmath}
\usepackage{centernot}
\usepackage{adjustbox}
\usepackage{amssymb,amscd}
\usepackage{comment}
\usepackage{amsmath,amsthm}
\usepackage{marginnote}

\newcommand{\ot}{\otimes}
\newcommand{\End}{\mbox{End}}
\newcommand{\unit}{\boldsymbol{1}}
\newtheorem{definition}{Definition}[section]
\newtheorem{theorem}{Theorem}[section]
\newtheorem{Remark}{Remark}

\newtheorem{proposition}{Proposition}[section]

\def\be{\begin{equation}}
\def\ee{\end{equation}}

\def\IC{\mathbb{C}}

\def\IH{\mathbb{H}}

\def\IR{{\mathbb{R}}}

\def\IZ{{\mathbb{Z}}}

\def\CA{{\cal A}}
\def\CB{{\cal B}}

\def\CD {{\cal D}}

\def\CW {{\cal W}}
\def\CX {{\cal X}}

\def\CZ {{\cal Z}}

\def\CB {{\cal B}}
\def\CQ {{\cal Q}}

\def\CX {{\cal X}}
\def\CY {{\cal Y}}
\def\CZ{{\cal Z}}
\def\la{\langle}
\def\ra{\rangle}

\def\one{{\hbox{ 1\kern-.8mm l}}}

\def\be{\bar{e}}

\def\cl{C\ell}

\def\cl{C\ell} 
\def\cl{C\ell}

\def\be{ \begin{equation} }
\def\ee{ \end{equation}}

\def\Aut{{\rm Aut}}

\title{Dyson's Classification And Real Division Superalgebras}
\author{Roman Geiko and Gregory W.~Moore}
                                           \affiliation{Department of Physics and Astronomy, Rutgers University ,\\
                                           126 Frelinghuysen Rd., Piscataway NJ 08855, USA}
\emailAdd{roman.geiko@physics.rutgers.edu} \emailAdd{gmoore@physics.rutgers.edu}
\abstract{It is well-known that unitary irreducible representations of groups can be usefully classified in a 3-fold classification scheme: Real, Complex, Quaternionic. In 1962 Freeman Dyson pointed out that there is an analogous 10-fold classification of irreducible representations of groups involving both unitary and antiunitary operators. More recently, it was
realized that there is also a 10-fold classification scheme involving superdivision algebras. Here we give a careful proof of the equivalence of these two 10-fold ways.
\newline
\newline
\arxivnumber{2010.01675}}
\begin{document}
\maketitle

\section{Introduction And Conclusion}

A fundamental theorem of E. Wigner \cite{Wigner1931, Bargmann1964, Freed:2011, Weinberg:1995} relates the theory of symmetry
in quantum mechanical systems to group representation theory.
Complex linear representations can be characterized as
real, complex, and quaternionic (a.k.a. pseudoreal).
These distinctions can play an important role in physical
applications of representation theory. One example is Kramers'
theorem: From a mathematical viewpoint, it is simply the
observation that the complex dimension of a quaternionic representation is even. A second 
example arises when studying index theory and the closely related subject of anomalies in quantum field theories with fermions. 
In this case, the nature of the
fermionic representation of the gauge group plays an important role.  
Wigner's theorem actually allows for symmetry transformations which are both unitary and
antiunitary. In 1962, when studying symmetry properties
of Hamiltonians in nuclear physics F. Dyson pointed out that
the three-fold classification of complex linear representations
should be generalized to a ten-fold classification of
representations that contain both unitary and antiunitary
symmetry transformations \cite{Dyson_1962}.

In more recent times, several other ten-fold classifications of aspects of symmetry in
quantum mechanics have appeared. One of these, described in \cite{Fidkowski:2009dba,Freed:2013}, is based on the 10-fold
classification of real superalgebras established by C.T.C. Wall \cite{Wall1964}.
It was suggested in \cite{Freed:2013} that the Dyson classification
and the Wall classification are ``essentially the same,'' although
the precise relation between the two schemes was imperfectly understood.
It is the purpose of this paper to fill in this gap and clarify
the relation. We will explain, rather thoroughly and carefully,
that they are, indeed, essentially the same. We must
stop and savor the delicious irony that the English physicist Dyson and
the English topologist Wall were establishing their results at the same
time, but
  the connection between their work only became apparent half a century
later.

We should mention in passing some of the other 10-fold
classification schemes which have appeared in the literature on symmetry
and quantum mechanics.  Foremost among these is the renowned ``Altland-Zirnbauer'' scheme for classification of free Fermion Hamiltonians \cite{Atland1997,Freed:2013,Heinzner:2004xj,
Zirnbauer2013, Zirnbauer:2020nab}.
Another classification scheme called a ``relativistic ten-fold way'' was
put forward by Freed and Hopkins in \cite{Freed:2016rqq}. In quantum
systems that come with a graded Hilbert space, the symmetry groups actually have a
natural $\IZ_2 \times \IZ_2$ grading as stressed in \cite{Freed:2013,Moore:2014}. 
One $\IZ_2$ is provided by the unitary/antiunitary
dichotomy and the other by the grading. In this case there is a similar
10-fold classification scheme based on superalgebras similar to that discussed here. The AZ classification scheme is a special case of this more general statement. The generalization of the AZ classification scheme in principle applies to bosons as well as fermions, interacting or
not, as emphasized in \cite{Moore:2014}. The ensembles in which the Hamiltonian is an odd element will have the special property that the energy spectrum has an inversion symmetry around zero.

Section two of this paper reviews the representation theory of groups
involving
both unitary and antiunitary operators. Two categories of representations
are studied and we discuss how the Schur lemma leads to 3-fold and
10-fold classifications of representations in the two cases.
In section three we review Dyson's 10-fold classification. In
section four we formulate and prove our main theorem, Theorem \ref{conj}, that
relates Dyson's classification to superalgebras. The main
result is summarized by the one-one correspondence given by the
first two columns of Table \ref{table}. Section five comments on an analog of the
Frobenius-Schur indicator for the 10-fold way. It is closely related to
recent work of Ichikawa and Tachikawa \cite{ichikawa2020super}. Section six demonstrates
how, starting from complex representations of a group $G$, one can
construct all 10 Dyson Types. Appendix \ref{B} summarizes basic facts about
superalgebras. Appendix \ref{Clifford algebras} summarizes our notation for Clifford algebras.
Finally, for convenience, we give a reprise of the statement and
proof of Wall's classification, following a treatment by P. Deligne.

We would like to call attention to two other closely related papers.
In \cite{neeb2017} $\phi$-representations were studied in relation  to modular theory of operator algebras. While completing our draft we learned of a very closely related paper by D. Rumynin and J. Taylor \cite{rumynin2020real}, which contains a modern review of Dyson's classification.

\section*{Acknowledgements}
We thank Dan Freed, Alexei Kitaev, Yuji Tachikawa, and Martin Zirnbauer for helpful conversations and correspondence. R.G. thanks Andrei Grekov for helpful discussions. G.M. had the great good fortune to have known Freeman Dyson. It was a pleasure discussing with him
the problem which is solved in this paper.
The authors are supported by the DOE under grant DOE-SC0010008 to Rutgers.

\section{Review Of \texorpdfstring{$\phi$}{TEXT}-representation Theory}
\label{Sec2}
We wish to define a complex semilinear action of a compact group $G$ on a complex vector space. By semi-linearity we mean that the group elements act either $\IC$-linearly or $\IC$-antilinearly. In order to   do that
conveniently, we ``decomplexify'' our complex vector spaces by presenting them as real vector spaces equipped with a complex structure $V:=(V_{\IR},I)$.
Below, we give basic definitions of $\phi$-representations (also known as antiunitary representations, semilinear representations, and co-representations), their subrepresentations etc.
\begin{definition}
\label{phi-rep-Def}
Let $G$ be a compact $\mathbb{Z}_2$-graded group, where the grading is given by a group homomorphism
\begin{equation}
\label{gradingG}
    \phi: G \to \mathbb{Z}_2\;,
\end{equation}
which we assume to be surjective. Here $\IZ_2$ denotes the multiplicative group of square roots of unity.

1) A $\phi$-representation $(\rho, V_{\IR},I)$ of a $\mathbb{Z}_2$-graded group $(G,\phi)$ is a real vector space with a complex structure $(V_{\IR},I)$ and a homomorphism
\begin{equation}
\rho: G\to \End_{\;\IR}(V_{\IR})\;
\end{equation}
such that
\begin{equation}
    \rho(g)I=\phi(g)I\rho(g)\;,\quad
        \mbox{for}\quad  g\in G\;.
\end{equation}

2) A $\phi$-subrepresentation $(\rho', W_{\IR}, I_{W})$ of the $\phi$-representation $(\rho, V_{\IR}, I_V)$ is a $\phi$-representation such that $W_{\IR}$ is a real vector subspace of $V_{\IR}$ with $I_{V}|_{W}=I_{W}
$ and $\rho|_W=\rho'$.

3) A $\phi$-representation is called irreducible if it has no non-trivial $\phi$-subrepresentations.
\end{definition}

The set of $\phi$-reps of the graded group $G$ can be seen as objects of some category. In the next two sections, we define two different categories consisting of $\phi$-reps with different types of morphisms among them.

\subsection{The Category \textbf{Rep}\texorpdfstring{$_{\IC}(G,\phi)$}{TEXT}}
Here, we define the category \textbf{Rep}$_{\IC}(G,\phi)$ of $\phi$-representations of the compact graded group $G$ with $\IC$-linear $G$-equivariant morphisms.

\begin{definition}
\label{defhomC}
 A complex-linear morphism (intertwiner) between two $\phi$-representations $(\rho',V_{\IR}',I')$ and $(\rho'',V_{\IR}'',I'')$ is a complex-linear map $f\in \mbox{Hom}_{\IC}(V',V'')$ commuting with the $G$-action
\begin{align}
\label{c-lin.morphism1}
    f\rho'(g)&=\rho''(g)\;f\;,\quad \forall g\in G\;,\\
\label{c-lin.morphism2}
    fI'&=I''f \;.
\end{align}
\end{definition}
We denote the set of such morphisms as $\mbox{Hom}_{\IC}^{G}(V',V'')$.

Since the group $G$ under consideration is compact, all representations are fully reducible. Indeed, the compactness of $G$ allows us to make all the $\rho(g)$ unitary or antiunitary. In other words, there always exists a $\phi$-equivariant inner product, which allows us to define orthogonal complements to irreducible subrepresentations. Thus, we can always decompose any $\phi$-representation into a direct sum of irreducible ones. In other words, we can say that \textbf{Rep}$_{\IC}(G,\phi)$ is a semisimple category.

The morphisms between any pair of simple objects are described by Schur's lemma.

\begin{theorem}{(Schur's lemma for $\phi$- representations)}
\label{phiSchur}

1. Let $f$
 be a complex-linear intertwiner between two $\phi$-irreps. Then, $f$ is either the zero map or an isomorphism.

2. The set of complex-linear self-intertwiners of a $\phi$-irrep, $\End\;^{G}_{\IC}(V)$, has the structure of a real division algebra.
\end{theorem}
\begin{proof}
1. Let us notice that $\mbox{ker}(f)$ and $\mbox{Im}(f)$ are $\phi$-subrepresentations. Since we are looking for a map between irreducible representations, the map can be either an isomorphism or the zero-map.

2. In general, $\mbox{Hom}^{G}_{\IC}(V',V'')$ is a real vector space. Indeed, assume we found $f$ satisfying \eqref{c-lin.morphism1} and \eqref{c-lin.morphism2}. Our $\phi$-irreps always contain at least one $\IC$-antilinear $\rho(g)$. For such $g$:
\begin{equation}
    f\rho'(g)=\rho''(g)\;f\quad \centernot \implies f I' \rho'(g)=\rho''(g) I''f \;.
\end{equation}
So, $fI'$ and $I''f $ do not necessarily belong to $\mbox{Hom}^{G}_{\IC}(V',V'')$.

It is clear that $\End_{\IC}^{G}(V)$ has the structure of a division algebra. This is a vector space with a multiplication - composition of maps. Due to the first point, every non-zero element of this algebra is invertible.
\end{proof}

A theorem due to Frobenius \ref{FrobTh} lists all real (unital associative) division algebras: they are isomorphic to $\IR$, $\IC$ or $\IH$. Then, a self-intertwiner of an irreducible $\phi$-representation acts as one of the real division algebras on the underlying representation space.
The fact that there are precisely three real associative unital division algebras underlies the number three in Dyson's famous Three-Fold Way.

\subsection{The Category \textbf{Rep}\texorpdfstring{$_{\IR}(G,\phi)$}{TEXT}}
\label{RepR}
Here, we define a category with morphisms being general real-linear maps, and discuss its basic properties. In analogy with morphisms among graded rings, which might shift the grading, we introduce the representation category with morphisms being semilinear maps.

As defined, $\phi$-representations always refer to real vector spaces endowed with complex structures. The most general morphisms between real vector spaces are real linear maps. Using the complex structures on both source and target spaces $V'=(V_{\IR}'
,I')$ and $V''=(V_{\IR}'',I
'')$, we can define an involution on the space of $\IR$-linear maps:
\begin{equation}
    f \mapsto -I''f I'\;, \quad f\in \mbox{Hom}_{\IR}(V'_{\IR},V''_{\IR}).
\end{equation}
There exists a homogeneous basis, such that even morphisms are $\IC$-linear and odd morphisms are $\IC$-antilinear maps correspondingly. This defines a super vector space structure on $\mbox{Hom}_{\IR}(V',V'')$. Since the composition of morphisms respects the $\IZ_2$-grading, it follows that $\End_{\IR}(V_{\IR}):=\mbox{Hom}_{\IR}(V_{\IR},V_{\IR})$ has the structure of a real superalgebra. For any real vector space $V_{\IR}$ with the complex structure $I$, we denote the superalgebra of $\IR$-linear maps graded by the adjoint action of $I$ as $\End_{\IR}(V_{\IR},I)$.

Having this observation at hand, we can define morphisms between $\phi$-reps to be $\IR$-linear super G-equivariant maps. Homogeneous $f^{0}, f^{1}\in \mbox{Hom}_{\IR}^{G}(V',V'')$ satisfy:
\begin{align}
f^{0}I'&=+I''f^{0}\;,\\
f^{0}\rho'(g)&=+\rho''(g)f^{0}\;,\\
    f^{1}I'&=-I''f^{1}\;,\\
    f^{1}\rho'(g)&=\phi(g)\rho''(g)f^1\;.
\end{align}

By $\textbf{Rep}_{\IR}^{G}(G,\phi)$ we denote the category of $\phi$-representations with $\IR$-linear super $G$-equivariant maps as morphisms.

 \begin{theorem}{(super-Schur's lemma for $\phi$- representations)}

1. Let $f$ be a homogeneous $\IR$-linear morphism between $\phi$-irreps. Then, $f$ is  either  the  zero  map  or  an  isomorphism.

2. The set  of  $\IR$-linear  endomorphisms $\End^{G}_{\IR}(V_{\IR},I)$ of a $\phi$-irrep $(V_{\IR},I)$, has  the  structure  of  a  real division superalgebra.
\begin{proof}
The proof of the first point is literally the same as for the Schur's lemma \ref{phiSchur}. Point two follows from what we said above.
\end{proof}
\end{theorem}

Due to a theorem by Wall \cite{Wall1964} (see Appendix \ref{Brauer} for a review), there are 10 real division superalgebras, so that any simple object of $\textbf{Rep}_{\IR}(G,\phi)$ has one of the ten types of endomorphisms. If we endow the category of $\phi$-representations with real-linear morphisms, we get the 10-fold classification of non-trivial morphisms among irreducible $\phi$-reps. Let us denote an endomorphism superalgebra of the given $\phi$-irrep as $\CW$. The set of real division superalgebras can be identified with the graded Brauer monoid of reals $sBR(\IR)$, considered as a set. The latter is as a set,
\begin{equation}
sBR(\IR)=\{\cl_{-3},\;\cl_{-2},\;\cl_{-1},\;\cl_{1},\;\cl_{2}, \;\cl_{3},\; \;\IR\;, \IH,\; \IC,\; \IC\ell_{1}\}\;.
\end{equation}
See Appendix \ref{Clifford algebras} for the notation. Note we consider superalgebras isomorphic if they are isomorphic in the graded sense: an isomorphism must be grading-preserving. Thus, we have a map from $\phi$-irreps to $sBR(\IR)$
\begin{equation}
    (\rho,V_{\IR},I)\mapsto [\CW]\;.
\end{equation}
We denote this map by $\delta$.

\section{Review of Dyson's Classification of \texorpdfstring{$\phi$}{TEXT}-irreps}
\label{Sec3}
Here, we review some aspects of theory of $\phi$-representations discussed by Dyson in his seminal paper.  In order to obtain a coarse classification of irreducible $\phi$-representations of $(G,\phi)$, we shall compress the initial data $(\rho, V_{\IR}, I)$ into a triple of real associative algebras ($\CD$, $\CB$, $\CA$).

Following Dyson, we define triples of algebras $(\CD,\CB,\CA)$, and list all possible triples potentially arising from an arbitrary $\phi$-irrep. Since $V_{\IR}$ has a complex structure, it has even real dimension; let us denote $2n=\mbox{dim}(V_{\IR})$.

\begin{definition}
\label{Dyson algs}
Let $(\rho, V_{\IR}, I)$ be a $\phi$-irrep of a graded group $(G,\phi)$.

1) We denote the real algebra generated by $\rho (g)$ for $g \in G$ such that $\phi(g)=1 $ as $\CA\subset \End_{\IR}(V_{\IR})$, and by $\mathcal{X}$ we denote its commutant within $\End_{\IR}(V_{\IR})$.

2) The real algebra obtained by joining $I$ to $\CA$ is denoted by $\CB$, and $\CY$ is its commutant within $\End_{\IR}(V_{\IR})$.

3) The real algebra generated by $\rho(g)$ for all $g \in G$ and $I$ is denoted by $\CD$, and $\CZ$ is its commutant within $\End_{\IR}(V_{\IR})$.
\end{definition}

First of all, we notice that irreducibility of $\rho$ is equivalent to simplicity of $\CD$. Then, $\CD$ is a real simple algebra and, due to a  theorem by Wedderburn, is isomorphic to the full matrix algebra with real, complex, or quaternionic coefficients. It is convenient to introduce the following notation: we denote the simple real matrix algebra with coefficients in a real division algebra  $\End_{\IR}(\IR^{n})\ot_{\IR} \mathbb{K}\cong \End_{\IR}(\mathbb{K}^{n})$ by $\mathbb{K}(n)$. Also, we will face algebras made of $m$ copies of $\mathbb{K}(n)$, we denote them by $m \mathbb{K}(n)\cong \unit_m\ot\mathbb{K}(n)$. Explicitly, $m \mathbb{K}(n)$ consists of matrices of the form
\begin{equation}
     \begin{pmatrix}
\mathbf{M}& &  & &  \\
 &  &\ddots &  & & \\
 &  & &  & \mathbf{M}
\end{pmatrix}\;,\quad \mbox{with}\; \mathbf{M}\in \mathbb{K}(n)\;,
\end{equation}
and as an algebra, $m \mathbb{K}(n)\cong \mathbb{K}(n)$. Using this notation, we present three alternatives
\begin{equation}
    \CD \cong \IR(2n)\;,\quad \CD\cong \IC(n)\;,\quad \CD=\IH(m)\;,
\end{equation}
where $n=2m$. Using the double commutant theorem, we can find the commutants of those algebras within $\IR(2n)$. With those three alternatives, the commutants respectively are
\begin{equation}
    \CZ \cong 2n\IR\;,\quad \CZ\cong n\IC\;,\quad \CZ=m\IH\;.
\end{equation}
Clearly, these three algebras coincide with the algebras appearing in Schur's lemma for irreps from $\textbf{Rep}_{\IC}(G,\phi)$.

Now we proceed to the possible types of $\CB$ with having $\CD$ fixed. The real dimension of $\CD$ is twice that of $\CB$, since $G$ has non-trivial $\mathbb{Z}_2$-grading. In other words, multiplication by $\rho(g)$ for any $g$ with $\phi(g)=-1$ defines an isomorphism from $\CB$ to a complement of $\CB$ within $\CD$. While $V$ need not be an irrep of $\CB$, it will be completely reducible.

Let us analyze the general structure of a unitary subrepresentation within the given $\phi$-irrep. Denote the map $\rho$ restricted to $G^{0}=\{g\in G\;| \phi(g)=1\}$ as $\rho^{\IC}$, since we refer to a standard $\IC$-linear representation. So, over the complex numbers, $\rho^{\IC}$ can be reduced to $N$ blocks:
\begin{equation}
\label{rhoCdecomposed}
    \rho^{\IC}(u)=\rho^{\IC}_1(u)\oplus \dots \oplus\rho^{\IC}_{N}(u)\;,\quad \mbox{where} \quad u\in G^{0}\;.
\end{equation}
For any odd element $a$ of the group $G$, we can define a map
\begin{equation}
    V_{a}(u)=a^{-1}\;u\;a\;,
\end{equation}
which is an automorphism of $G^{0}$ and consequently
\begin{equation}
\label{Vaction}
    \rho^{\IC}\mapsto \rho(a)^{-1} \rho^{\IC}\; \rho(a)
\end{equation}
is an automorphism of $\rho^{\IC}$. The map \eqref{Vaction} has to exchange the components of \eqref{rhoCdecomposed} in order for $\rho$ to be irreducible. In other words, \eqref{Vaction} acts as a non-trivial permutation on the components of \eqref{rhoCdecomposed}. In order to proceed further, let us apply \eqref{Vaction} twice
\begin{equation}
\label{Vaction2}
    \rho^{\IC}\mapsto \rho(a^2)^{-1}\;\rho^{\IC}\; \rho(a^2)=\rho^{\IC}(w)^{-1}\;\rho^{\IC}\; \rho^{\IC}(w)\;,\quad \mbox{where} \quad w\in G^{0}\;,
\end{equation}
which gives us an automorphism of the unitary subrepresentation, and leaves the components of \eqref{rhoCdecomposed} invariant. This shows that $\eqref{Vaction}$ defines an involution on $\rho^{\IC}$. This involution either acts trivially on a component of \eqref{rhoCdecomposed}, or transposes a pair of components, so that irreducible unitary subrepresentations form either singlets or doublets under the action of $\eqref{Vaction}$. A crucial observation is that each singlet, as well as each doublet, corresponds to precisely one irreducible $\phi$-representation.

\begin{proposition}
\label{rhoC}
Unitary representation $\rho^{\IC}$ contains either one or two irreducible components; if there are two components, then \eqref{Vaction} exchanges them.
\end{proposition}
We can now show that, given knowledge of $\CD$, one can deduce the algebra $\CB$. Note that $\CB$   is a complex semisimple algebra containing up to two blocks of the same dimension (since \eqref{Vaction} is an invertible map). Also, as follows from the consideration above, the simple components have to be isomorphic algebras. Using what we said above and the fact that the real dimension of $\CD$ is twice that of $\CB$, we can fix the complex algebra $\CB$ uniquely (up to isomorphism):
\begin{equation}
\label{DB}
 \begin{aligned}
\CD& \cong \mathbb{R}(2n) \quad  \Rightarrow \quad  \CB \cong \IC(n)\;,\\
\CD& \cong \IC(n) \;\;\quad  \Rightarrow \quad  \CB \cong \IC(m)\oplus \IC(m)\;,\quad \quad m=n/2\;,\\
\CD& \cong \IH(m) \; \quad  \Rightarrow \quad  \CB \cong 2\IC(m)\;.
 \end{aligned}
 \end{equation}
  
  We now proceed to $\CA$: this is an algebra such that adjoining $I$ to it gives $\CB$. 
  The simple components are simple real algebras. If $I$ belongs to $\CA$, then its order is $2n^2$, otherwise it is $n^2$. Assuming $\CB$ has one simple component:
\begin{equation}
\label{BtoA}
\CB \cong \IC(n) \quad  \Rightarrow \quad
\begin{cases}
 \quad\CA \cong 2\IR(n)\;,\\
 \quad\CA \cong \IC(n)\;,\\
 \quad\CA \cong \IH(m)\;.
\end{cases}
 \end{equation}

 This result can be directly applied to the first and third line in \eqref{DB}. For the second line, the result \eqref{BtoA} should apply to the two components independently. As follows from the derivation of Proposition \ref{rhoC}, there are automorphisms of the full group algebra permuting its simple components. Rephrased, $\CB$ might have up to two simple components being isomorphic as algebras. Matching the dimensions, we get the following possibilities:
 \begin{equation}
\label{CCtoA}
\CB=\IC(m)\oplus\IC(m) \quad  \Rightarrow \quad
\begin{cases}
 \quad\CA=2\IR(m)\oplus2\IR(m)\;,\\
 \quad\CA=\IC(m)\oplus\IC(m)\;,\\
 \quad\CA=2\IC(m)\;,\\
 \quad\CA=\IH(p)\oplus\IH(p)\;,\\
\end{cases}
 \end{equation}
 where $p=n/4$. According to the standard terminology, we say that the complex subrepresentation $\rho^{\IC}$ is potentially real, truly complex, or potentially quaternionic if $\CA$ has real, complex, or quaternionic type respectively.

 This finishes the list of possible triples $(\CD,\CB,\CA)$ arising as Dyson's algebras \ref{Dyson algs} of the $\phi$-irrep, and any triple of algebras from this table may correspond to some $\phi$-irrep. Let us call a   triple of algebras from this table a \textit{Dyson triple}. It is clear that we can introduce an equivalence relation on the Dyson triples: we say that two triples are similar if the constituent algebras differ by a common tensor factor $\IR(p)$ for some $p$. Then, we obtain ten similarity classes of the Dyson triples. Similarity classes of the Dyson triples are called \textit{Dyson Types}. Recall, $\CB$ is not independent, so that $(\CD,\CB,\CA)$ contains the same data as $(\CD,\CA)$. According to that, any Dyson Type can be encoded by the similarity classes of algebras $\CD$ and $\CA$. Following Dyson, we denote the similarity class of $(\CD,\CA)$ as $[\CD][\CA]$, and for brevity $[\IC][\IC\oplus \IC]\equiv \IC\IC_1$, and $[\IC][\IC]\equiv \IC\IC_2$, where square brackets stand for the equivalence class defined above.

 \begin{table}[h!]
\begin{center}
\begin{adjustbox}{width=1\textwidth}
\begin{tabular}{ |c|c|c|c|c|c|c|c| }
\hline
$\mathbb{D} $ & $Type$ & $\CD$ & $\CB$ & $\CA$ & $\CX$ & $\CY$ & $\CZ$ \\
  \hline
 $\cl_{1}$ & $\mathbb{RR}$ & $\mathbb{R}(2n)$ & $\IC(n)$ & $2\mathbb{R}(n)$  & $n\IR(2)$ & $n\IC$ & $2n \IR$\\
 \hline
$\IR$ & $\mathbb{RC}$  & $\IR(2n)$ & $\IC(n)$ & $\IC(n)$ & $n\IC$ & $n\IC$ & $2n\IR$\\
\hline
$\cl_{-1}$ & $\mathbb{RH}$  & $\IR(4m)$ & $\IC(2m)$ & $\IH(m)$ & $m \IH^{\scriptstyle{opp}}$ & $2m \IC$ & $4 m \IR$\\
\hline
$\cl_{2}$ & $\mathbb{CR}$  & $\IC(2m)$ & $\IC(m)\oplus \IC(m)$ & $2\IR(m) \oplus 2\IR(m)$ & $m \IR(2) \oplus m \IR(2)$ & $m \IC \oplus m \IC$ & $2 m \IC$\\
\hline
$\IC$ & $\mathbb{CC}_1$  & $\IC(2m)$ & $\IC(m)\oplus \IC(m)$ & $\IC(m)\oplus \IC(m)$ & $m \IC \oplus m \IC$ & $m \IC \oplus m \IC$ & $2 m \IC$\\
\hline
$\IC\ell_1$ & $\mathbb{CC}_2$  & $\IC(2m)$ & $\IC(m)\oplus \IC(m)$ & $2\IC(m)$  & $m \IC(2)$ & $ m \IC\oplus m \IC$ & $2m \IC$ \\
\hline
$\cl_{-2}$ & $\mathbb{CH}$  & $\IC(4p)$ & $\IC(2p)\oplus \IC(2p)$ & $\IH(p) \oplus \IH(p)$  & $p \IH^{\scriptstyle{opp}} \oplus p \IH^{\scriptstyle{opp}}$ & $ 2p \IC\oplus 2p \IC$ & $4p \IC$ \\
\hline
$\cl_{3}$ & $\mathbb{HR}$  & $\IH(m)$ & $2\IC(m)$ & $4\IR(m)$  & $m\mathbb{R}(4)$ & $m\IC(2)$ & $m \IH^{\scriptstyle{opp}}$ \\
\hline
$\IH$ & $\mathbb{HC}$  & $\IH(m)$ & $2\IC(m)$ & $2\IC(m)$  & $m\IC(2)$ & $m\IC(2)$ & $m \IH^{\scriptstyle{opp}}$ \\
\hline
$\cl_{-3}$ & $\mathbb{HH}$  & $\IH(2p)$ & $2\IC(2p)$ & $2\IH(p)$  & $p\IH^{\scriptstyle{opp}}$ & $2p\IC(2)$ & $2p\IH^{\scriptstyle{opp}}$ \\
\hline
\end{tabular}
\end{adjustbox}
\caption{The real division superalgebras matched with the Dyson Types and their Dyson triples. The first column will be explained below - it is one the main points of this paper.}
\label{table}
\end{center}
\end{table}

Examples of $\phi$-irreps can be found in the original paper \cite{Dyson_1962}. In particular, Dyson shows that all ten possibilities indeed occur. In section \ref{HowToConstruct} below, we discuss how $\phi$-irreps arise from ordinary complex-linear irreps.

We have defined a map from $\phi$-irreps to the Dyson triples.
\begin{equation}
(\rho, V_{\IR},I)\to (\CD,\CB,\CA)\;.
\end{equation}
Further, taking the similarity class of the Dyson triple defines a map to the Dyson Type:
\begin{align}
\alpha:\quad \mbox{Irreps}(G,\phi)&\to Dyson\; Type\;,\\
(\rho, V_{\IR},I)&\mapsto [\CD][\CA]\;.
\end{align}

\section{Real Division Superalgebras And Dyson Types}
\label{Sec4}
We have seen that there is a correspondence between the three possible similarity types of $\CD$ and three real division algebras. It is natural to ask - is there a correspondence between the ten Dyson Types and ten real division superalgebras? This question was first raised in \cite{Freed:2013}, and formulated as a conjecture about a 1-1 correspondence between the two. Some examples are worked out in \cite{Moore:2013}. One of the main goals of this paper is to show that the classifications by Dyson Types and graded commutants are essentially the same. We formulate this as a theorem and prove it below.

 \begin{theorem} \label{conj}
 There is a 1-1 correspondence between the Dyson Types and real division superalgebras, such that for any given $\phi$-irrep of any $\mathbb{Z}_2$-graded group $(G,\phi)$, the endomorphism superalgebra $\CW$ for this irrep corresponds to the Dyson Type extracted from this irrep.
\end{theorem}

For any $\phi$-irrep, we can determine its Dyson Type and the endomorphism superalgebra $\CW$ separately. We shall prove that, for any $(G, \phi)$ and any $\phi$-irrep the two are governed by the 1-1 correspondence of ten objects in the first two columns of Table \ref{table}.

\subsection{The Road Map}
We prove the theorem in two steps. Firstly, for any $\phi$-irrep we introduce the superalgebra $
\CA_{s}$ as a superalgebra generated by $\rho(g)$ for all $g\in G$, where the grading is inherited from that of $G$. We will prove that for any $\phi$-irrep, $\CA_{s}$ and the Dyson triple extracted from this $\phi$-irrep are in 1-1 correspondence, so that there is a map \begin{equation}
    Dyson \;triples \to Superalgebras\;.
\end{equation}
Where, $Superalgebras$ stands for the set of superalgebras $\CA_{s}$  which can arise from  $\phi$-irreps. Note, there are restrictions on $\CA_{s}$ dictated by the fact they occur from representations of graded groups, so they are not general superalgebras. In the previous section, we introduced the equivalence relation of having a common tensor factor $\IR(n)$; under this equivalence Dyson triples descend to Dyson Types, and $Superalgebras$ descend to $GroupSAlg$. There is an isomorphism of sets:
\begin{align}
\label{beta}
   \beta:\quad  Dyson \;Types &\to GroupSAlg\;,\\
   Dyson\; Type &\mapsto [\CA_{s}]\;.
\end{align}
where $GroupSAlg$ stands for the set of equivalence classes of the superalgebras $\CA_{s}$. It turns out that $GroupSAlg$ contains all the ten classes from $sBR(\IR)$.

The next step will be to calculate the graded commutant of the superalgebra $\CA_{s}$ within $\End_{\IR}(V_{\IR},I)$. This graded commutant is nothing but an endomorphism superalgebra of the $\phi$-irrep. Taking the graded commutant is an invertible operation, which defines an isomorphism of sets
\begin{align}
\label{gamma}
   \gamma:  GroupSAlg &\to sBR(\IR)\;\\
   [\CA_{s}]&\mapsto [\CW]\;.
\end{align}

Our goal is to end up with an isomorphism of sets $Dyson \;Types\to sBR(\IR)$ not depending on the specific group and representation. We achieve it by constructing the isomorphism $\gamma \circ \beta$ without referring to any specific $\phi$-irrep. We analyze what are the algebraic conditions the $\phi$-irrep has to satisfy in order to have some fixed Dyson Type. Those conditions allow us to fix the corresponding superalgebra $\CA_{s}$ up to equivalence. In its turn, any simple superalgebra has a unique graded commutant within $\End_{\IR}(V_{\IR},I)$. In fact, we can go other way unambiguously, such that $\gamma \circ \beta$ is an isomorphism of sets.

\subsection{Group Superalgebras}

As we discussed above, the full matrix algebra $\End_{\IR}(V_{\IR})$ admits a $\IZ_2$-grading given by the adjoint action of the complex structure. Then, the $\phi$-representation $(\rho, V_{\IR},I)$ can be seen as a grade-preserving homomorphism $\rho: G\to \End_{\IR}(V_{\IR},I)$. This morphism can always fit into the commutative diagram:
\begin{center}
\begin{tikzcd}
  (G,\phi) \arrow{dr}  \arrow{rr}{\rho} &                         & \End_{\IR}(V_{\IR},I)\\
  &  \IR^{\phi}[G]  \arrow{ur}{\psi} &
\end{tikzcd}
\end{center}
where $\IR^{\phi}$ is the group superalgebra (a super vector space with the basis given by the elements of $(G,\phi)$), and $\psi$ is an $\IR$-linear grade-preserving injective map. Let us denote  the image of $\psi$ as $\CA_{s}$. Alternatively, $\CA_{s}$ is an $\IR$-linear span of $\{\rho(g)|\;g\in G\}$. The superalgebra $\CA_{s}$ is a close analog of Dyson's algebras associated to the $\phi$-irrep. The even part of $\CA_{s}$ is nothing but the Dyson's $\CA\equiv \CA_{s}^{0}$, and $\CA_{s}^{1}$ is the image of the odd elements $\CA_{s}^{1}\equiv \{\rho(g)|\;g\in G,\; \phi(g)=-1\}$.

Let us notice that the largest of the Dyson algebras $\CD$ contains $\CA_{s}^{\scriptstyle{ungr}}$ together with an element which commutes with $\CA_{s}^{0}$ and anticommutes with $\CA_{s}^{1}$: $\CD = \la \CA_{s}^{\scriptstyle  {ungr}}, I\ra$, where angle brackets stand for additive and multiplicative closure.

As follows from definitions, $\CA_s$ is a simple superalgebra for any $\phi$-irrep. Since it arises from irreps of non-trivially graded groups, $\CA_{s}$ is not concentrated in the even degree. Moreover, there is an isomorphism of vector spaces $\CA_{s}^{0}\cong \CA_{s}^{1}$. We can always present this morphism as a multiplication by an odd element $u\in \CA_{s}^{0}$. Indeed, $G$ admits a coset decomposition $G= G^{0}\bigcup \tilde g G^{0}$, where $\tilde g$ is a coset representative; we can take $u=\rho(\tilde g)$, and it has the properties of a coset representative as well. Summarizing, there is an invertible element $u$ such that $\CA_{s}^{1}=u \CA_{s}^{0}$.

\subsection{Aside: \texorpdfstring{$\phi$}{TEXT}-Representation As Real-Linear Representations}
\label{Bosonization}
Let us clarify what we mean by the relation $\CD \equiv \la \CA_{s}^{ungr},I\ra$. Definition \ref{phi-rep-Def} of a $\phi$-representation of   $G$ can be recast as the definition of an ordinary representation of a group $\tilde G$ related to $G$. Indeed, $\tilde G$ will simply
be a semidirect product of $G$ with $U(1)$ (or $\mathbb{Z}_4$).

 Let us identify $\phi$ with the non-trivial homomorphism $\phi: G \to \mbox{Aut}(U(1))\cong \IZ_2$. The non-trivial generator of $\Aut(U(1))$ is the Galois action $\omega\to \omega^{-1}$. Using $\phi$, we can define a semidirect product $\tilde G\equiv U(1)\rtimes_{\phi} G$. It is clear from the definitions that any $\phi$-representation $(\rho, V_{\IR},I)$ of $(G,\phi)$ is equivalent to a real linear representation $(\rho, V_{\IR})$ of $\tilde G$. Further, we can introduce the group algebra $\IR[\tilde G]$: its image under $\rho$ will be nothing but $\CD$. This useful point of view was used to describe Dyson's classification in \cite{rumynin2020real}.

\subsection{Proof Of The Conjecture}
Firstly, we wish to convert the Dyson triples into superalgebras. With the $\CA_{s}$ defined above, we claim that the Dyson triples are equivalent to $\CA_{s}$.
\begin{proposition}
\label{DysonTripleToAs}
For any $\phi$-irrep, there is a 1-1 correspondence between the Dyson triple and the superalgebra $\CA_{s}$.
\end{proposition}

Any triple of algebras $(\CD, \CB, \CA)$ coming from a $\phi$-irrep implicitly contains information about $u$ because we can generate $\CD$ with the elements of $\CA$, $u$, and $I$. Let us recall that the Dyson Type contains Dyson triples differing only by a tensor factor of $\IR(n)$. Assume $\CD\cong \la \CA, u,I\ra$ such that $\CA_{s}\cong \CA\oplus u \CA$, then $\CD'\cong \CD\ot \IR(n)\cong \la \CA\ot \IR(n), u, I\ra$ corresponds to $\CA_{s}'\cong \CA_{s}\ot \IR(n)$. From here it follows that it is enough to consider the minimal Dyson triple: a Dyson triple from the Table \ref{table} having the parameter $n$ (or $m$, or $p$) equal to 1.

Our proof is by direct computation. Beginning with the minimal Dyson triple $(\CD,\CB,\CA)$, we reconstruct $u$. Having obtained $\CA_{s}$, we get back to a triple. Within the minimal irrep, we always pick the standard complex structure. As a real matrix acting on $V_{\IR}^{2q}$, where $q$ is the complex dimension of the minimal $\phi$-irrep, it always looks like $I=i\ot \unit_{q}$, where $i\in \IC$ is represented by the matrix $\bigl( \begin{smallmatrix}0 & -1\\ 1 & 0\end{smallmatrix}\bigr)$. We are looking for an operator $u$, anticommuting with $I$ such that $\CD\cong \la \CB, u\ra$. When working with $\IR(2)$, we always pick the ``quaternionic'' basis.

\paragraph{$(\IR\IR)\Longrightarrow \cl_{1}$.}
Here, $u$ is the real structure. It is an element of $\IR(2)$ squaring to $+\unit$ and anticommuting with $I$. Then, $\CA_{s}\cong \cl_{1}$.

\paragraph{$(\IR\IC)\Longrightarrow \cl_{2}$.}
Here, $\CA=2\IC$ and $\CB=\IC\oplus \IC \cong \IC\ot \la \unit, P\ra$, where $P$ is an abstract element of $\IR(2)$ squaring to $\unit$. In order to complete $\CB$ to $\CD$, we take $u=P'\ot P'$, where $P'\in \IR(2)$ is an operator squaring to $\unit$, and anticommuting with $i$ and $P$.  This shows $\CA_{s}=\cl_{2}$.

\paragraph{$(\IR\IH)\Longrightarrow \cl_{3}$.}
Given $\CA=\IH$ and $\CB=\IC(2)$, the complex structure $I=i'\in \IH^{\scriptstyle  {opp}}$. Then, $u=j'\in \IH^{\scriptstyle  {opp}}$. One can check that if $\CA_{s}^{1}= j'\IH$, then $\CA_{s}\cong \cl_{3}$.

\paragraph{$(\IC\IR)\Rightarrow \cl_{1,-1}$.}

Here, $\CA=2\IR\oplus 2\IR\cong 2\IR \otimes \la \unit, P\ra $. The complex structure $I=i\ot \unit$, and $u$ is an antilinear operator squaring to $+\unit$, let us call it $P'\ot P'$, where $P'\in \IR(2)$, is an operator anticommuting with $P\in \IR(2)$ and squaring to $+\unit$. This leads to $\CA_{s}\cong \cl_{1,-1}$

\paragraph{$(\IC\IC_1)\Longrightarrow \IC\ell_{2}$.}
Here, $\CA=\IC\oplus \IC\cong \IC \otimes \la \unit, P\ra $. The complex structure $I=i\ot  \unit$,  and $u$ is the same as in the previous case. If appropriately normalized, $u^2=\pm 1$ leads to $\CA_{s}\cong\IC\ell_{2}$.

\begin{Remark}
Even though, the superalgebras $\cl_{1,-1}$ and $\IC\ell_2$ are similar to $\IR$ and $\IC$ correspondingly, we keep writing them as they are. We are stressing that $\phi$ is a surjective map, and consequently $\CA_{s}$ is not a purely even superalgebra. Nothing prevents $\CA_{s}$ to be similar to a purely even superalgebra.
\end{Remark}

\paragraph{$(\IC\IC_2)\Longrightarrow \IC\ell_{1}$.}
Here, $\CA=2\IC \cong \unit \otimes \IC $. The complex structure $I=i\ot  \unit$,  and $u$ is the same as in the previous two cases. If appropriately normalized, $u^2=\pm 1$ leads to $\CA_{s}\cong\IC\ell _{1}$.

\paragraph{$(\IC\IH)\Longrightarrow \cl_{-4}$.}
Here, $\CA=\IH \ot \la \unit, P\ra$. The complex structure is $I=i'\ot \unit$ with $i'\in \IH^{\scriptstyle{opp}}$, and $u$ is a complex-linear combination of $\unit\ot P'$ and $\unit\ot i$. Then, $\CA_{s}^{1}= u \CA $ leads to $\CA_{s}\cong \cl_{4}$ (as well as to $\CA_{s}\cong \cl_{-4}$).

\paragraph{$(\IH\IR)\Longrightarrow \cl_{-1}$.}
In this case, $\CA=4\IR$, $\CB=2\IC$, such that $I=i$ and $u$ is any operator anticommuting with $I$ and squaring to $-\unit$, so is $u=j\in \IH$. Then, $\CA_{s}^{1}=\{j\}$, and $\CA_{s}\cong \cl_{-1}$.

\paragraph{$(\IH\IC)\Longrightarrow \cl_{-2}$.}
In this case, $\CA=2\IC$, $\CB=2\IC$, such that $I=i$ and $u=j\in \IH$. One may identify $\CA_{s}\cong \cl_{-2}$.

\paragraph{$(\IH\IH)\Longrightarrow \cl_{-3}$.}
Let us use the dimension argument. The superalgebra $\CA_{s}$ is a simple superalgebra of the real dimension $8$ (this dimension is twice the dimension of $\CA$). The candidates are the Clifford algebras of rank 3: $\cl_3$, $\cl_{-3}$, $\cl_{2,-1}\cong\cl_1 \ot \IR(2)$, and $\cl_{1,-2}$. Matching the even parts leaves us with $\cl_{3}$, and $\cl_{-3}$.

If $\CA\cong \IH$, the full algebra $\CD$, as a vector space, can be decomposed as $\CD= \IH\oplus I \IH\oplus u\IH\oplus uI \IH$. Then, depending on the class of $u^2$ in $\IR^{2}/\IR^{*}$, we get the two possibilities. If $u^2=-\unit$, then $\CD=\IH \ot \IH \cong \IR(4)$. If $u^2=\unit$, then $\CD=\IH \ot \IR(2)\cong \IH(2)$. Since we are interested in $\CD=\IH(2)$, the coset representative $u$ should satisfy $u^2=\unit$. Then, $\CA_{s}^{1}\cong u \IH$ implies $\CA_{s}\cong \cl_{-3}$.
\paragraph{}

Now we wish to obtain a Dyson triple from the superalgebra $\CA_{s}$. Recall, $\CD=\la \CA_{s}^{\scriptstyle  {ungr}}, I\ra$ such that $I$ defines the grading on $\CA_{s}$. So, the dimension of $\CD$ is twice that of $\CA_{s}^{\scriptstyle}$ in the case $I\in \CA_{s}^{0}$ and it is four times the dimension of $\CA_{s}^{0}$ if $I\notin \CA_{s}^{0}$. Summarizing, we say that $\CD$ is a simple algebra containing  $\CA_{s}^{\scriptstyle{ungr}}$ as a subalgebra or coinciding with the latter.

For any superalgebra, its even subalgebra and the algebra obtained by forgetting the grading is known. Let us recall some relevant cases in the following Table \ref{table2}.
\begin{table}[h!]
\begin{center}
\begin{adjustbox}{width=4.5cm}
\begin{tabular}{ |c|c|c|c|c|c|c|c| }
\hline
 $[\CA_{s}]$ & $[\CA_{s}^0]$ & $[\CA_{s}^{\scriptstyle  {ungr}}]$ \\
\hline
$\cl_{-4}$ & $ \IH \oplus \IH $  & $\IH(2)$  \\
\hline
$\cl_{-3}$ & $\IH$ & $\IH\oplus \IH$ \\
\hline
$\cl_{-2}$ & $\IC$ & $ \IH$ \\
\hline
$\cl_{-1}$ & $\IR$  & $\IC$  \\
\hline
$\cl_{1,-1}$ & $\IR\oplus \IR$ & $\IR(2)$ \\
\hline
$\cl_{1}$ & $\IR $ & $\IR\oplus \IR$ \\
\hline
$\cl_{2}$ & $\IC$ & $\IR(2)$ \\
\hline
$\cl_{3}$ & $ \IH$  & $\IC(2)$  \\
\hline
$\IC\ell_1$ & $\IC$ & $\IC\oplus \IC $\\
\hline
$\IC\ell_2$ & $\IC\oplus \IC$ & $\IC(2) $\\
\hline
\end{tabular}
\end{adjustbox}
\caption{The set of distinguished Morita-classes of superalgebras appearing form $\phi$-irreps.}
\label{table2}
\end{center}
\end{table}

\paragraph{$\cl_{1}\Longrightarrow (\IR\IR) $.} In this case, $\CA_{s}^{\scriptstyle{ungr}}=\IR\oplus \IR$. The order of $\CD$ is $4$, it is either $\IR(2)$, or $\IH$. Obviously, it cannot be $\IH$, as there is not enough elements squaring to $-\unit$. So, $\CD=\IR(2)$.

\paragraph{$\cl_{3}\Longrightarrow (\IR\IH)$.} The order of $\CD$ is 16, so it is either $\IH(2)$ or $\IR(4)$. We are completing $\IC(2)$ with an operator $I$ commuting with $\CA_{s}^{0}=\IH$. In this case, $\CA_{s}^{1}$ can be identified with $j'\in \IH^{\scriptstyle{opp}}$, where $j'$ is odd and anticommutes with $I$. Together, $I$ and $j'$ generate the full quaternion algebra $\IH^{\scriptstyle{opp}}$. So, $\CD=\IH\IH^{\scriptstyle{opp}}\cong \IR(4)$.

\paragraph{$\cl_{1,-1}\Longrightarrow (\IC\IR)$.} The order of $\CD$ is 8, so this is $\IC(2)$.

\paragraph{$\IC\ell_{1}\Longrightarrow (\IC \IC_2)$.} The order of $\CD$ is 8, so this is $\IC(2)$.

\paragraph{$\cl_{4}\Longrightarrow (\IC\IH) $.}
The order of $\CD$ is 32, so this is $\IC(4)$.

\paragraph{$\cl_{-1}\Longrightarrow (\IH\IR) $.}
The order of $\CD$ is 4, the candidates are $\IR(2)$ and $\IH$. Clearly, $\CD=\IH$.

\paragraph{$\cl_{-3}\Longrightarrow (\IH\IH) $.}
The order of $\CD$ is 16. The candidates are $\IH(2)$ and $\IR(4)$.  However, $\IR(4)$ does not contain $\CA_{s}^{\scriptstyle{ungr}}=\IH\oplus \IH$ as a subalgebra, while $\IH(2)$ clearly contains it.

This finishes the proof of proposition \ref{DysonTripleToAs}.

\subsection{Graded Commutants}
We have transferred data encoded in Dyson's analysis into the language of superalgebras. Thanks to that, computation of graded commutants of irreducible $\phi$-reps defined in \ref{RepR} becomes a conventional computation of graded commutants of group superalgebras within $\End_{\IR}(V_{\IR},I)$.

\begin{definition}
\label{grcom}
We define the graded commutant $\CZ_{s}(A,B)$ of a superalgebra $A$ within a superalgebra $B$ as the subset of $B$ containing elements of the form $w=w^0+w^1\in B$ such that
\begin{align}
\label{grcomm1}
   & w^{0}a\;=+a \;w^{0}\;,\\
   \label{wcomm2}
   & w^{1}a^{0}=+a^{0}w^{1}\;,\\
   \label{wcomm3}
    &w^{1}a^{1}=-a^{1}w^{1}\;.
\end{align}
In this notation, $\CZ_{s}(A)\equiv \CZ_{s}(A,A)$ is the supercenter of $A$.
\end{definition}

Matching this definition with that of $\CW$, we can identify the graded commutant of a $\phi$-irrep with the graded commutant of its group superalgebra within $\End_{\IR}(V_{\IR},I)$:
\begin{equation*}
    \End^{G}_{\IR}(V_{\IR},I)\equiv \CZ_{s}(\CA_{s},\End_{\IR}(V_{\IR},I))\;.
\end{equation*}

To proceed further, let us analyze the superalgebra $\End_{\IR}(V_{\IR},I)$. It is a real simple superalgebra with the supercenter $\IR$, so that it is similar to one of the real Clifford algebras. Picking the standard complex structure in $\End_{\IR}(V_{\IR})$, one can directly check the following isomorphisms:
\begin{equation}
    \End_{\IR}(\IR^{2},I)\cong \cl_2\;,\quad \End_{\IR}(\IR^{4},I)\cong \cl_{3,-1}\;,\quad \End_{\IR}(\IR^{8},I)\cong \cl_{4,-2}\;.
\end{equation}
There are two more isomorphisms: $\End_{\IR}(\IR^{4},I)\cong \cl_2\ot \IR(2)$ and $\End_{\IR}(\IR^{8},I)\cong \cl_{2}\ot \IR(4)$, which will be useful. In general, we define a grading on $\End_{\IR}(\IR^{2n},I)$ by the adjoint action of the standard complex structure $I=i\ot \unit_{n}$. Then, as a super vector space, $\End_{\IR}(\IR^{2n},I)\cong \la \unit,i\ra \ot \IR(n)\oplus \{P,P'\}\ot \IR(n)$, where we used the ``quaternionic'' basis for $\IR(2)$. Since we consider $\IR(n)$ as an ungraded algebra and the tensor product is the usual tensor product, $\End_{\IR}(\IR^{2n},I)$ is isomorphic to $\cl_{2}\ot \IR(n)$, so is similar to $\cl_2$. By the same argument,
\begin{equation}
    \End_{\IR}(\IR^{4n},I)\cong \cl_{3,-1}\ot \IR(n)\;,\; \End_{\IR}(\IR^{8n},I)\cong \cl_{4,-2}\ot \IR(n)\;.
\end{equation}

 Now, our aim is to compute graded commutants of superalgebras from the second column of the following Table \ref{table3} within the superalgebras from the third column. The resulting commutant $\CW$ is presented in the fourth column; we will calculate it shortly.
\begin{table}[h!]
\begin{center}
\begin{adjustbox}{width=8cm}
\begin{tabular}{ |c|c|c|c|c|c|c|c| }
\hline
$\mbox{Type}$ &  $\CA_{s}$ & $\End_{\IR}(V_{\IR},I)$ & $\CW$ \\
\hline
$\IR\IR$ &  $\cl_{1}\ot \IR(n)$  & $\cl_2\ot \IR(n)$ & $\cl_1$ \\
\hline
$\IR\IC$ & $\cl_{2}\ot \IR(n)$ & $\cl_2\ot \IR(n)$  & $\IR$ \\
\hline
$\IR\IH$  & $\cl_{3}\ot \IR(m)$ & $\cl_{3,-1} \ot \IR(n)$ & $\cl_{-1}$  \\
\hline
$\IC\IR$  & $\cl_{1,-1}\ot \IR(m)$ & $\cl_{3,-1}\ot \IR(m)$ &  $\cl_2$ \\
\hline
$\IC\IC_1$ & $\IC\ell_2\ot \IR(m)$ & $\cl_{3,-1}\ot \IR(m)$  & $\IC$\\
\hline
$\IC\IC_2$  & $\IC\ell_1\ot \IR(m)$ & $\cl_{3,-1}\ot \IR(m)$ & $\IC\ell_1$\\
\hline
$\IC\IH$  & $\cl_{4}\ot \IR(p)$ & $\cl_{4,-2}\ot \IR(p)$ & $\cl_{-2}$  \\
\hline
$\IH\IR$  & $\cl_{-1}\ot \IR(m)$ & $\cl_{3,-1}\ot \IR(m)$ & $\cl_{3}$  \\
\hline
$\IH\IC$ & $\cl_{-2}\ot \IR(m)$  & $\cl_{3,-1}\ot \IR(m)$ & $\IH$  \\
\hline
$\IH\IH$ & $\cl_{-3}\ot \IR(p)$ & $\cl_{4,-2}\ot \IR(p)$ & $\cl_{-3}$  \\
\hline
\end{tabular}
\end{adjustbox}
\caption{The Dyson Types matched with the superalgebras $\CA_{s}$.}
\label{table3}
\end{center}
\end{table}

First of all, let us notice that $\CZ_{s}(A\ot \IR(n),B\ot\IR(n))\cong \CZ_{s}(A,B)$, for any simple superalgebras $A$ and $B$, what is a consequence of the known isomorphism $\CZ(\IR(n),\IR(n))\cong\IR$. Then, we are left with a calculation of graded commutants of some Clifford algebras within other Clifford algebras. Recall that we define Clifford algebras through a vector space equipped by a symmetric bilinear form, see Appendix \ref{Clifford algebras}. The following proposition relates operations of taking commutants in the context of superalgebras with operation of taking orthogonal complement in the context of bilinear forms on vector spaces.
\begin{proposition}
\label{double-comm}
Let $(V,B)$ and $(W,B')$ be a pair of a real vector spaces equipped with symmetric bilinear forms such that $V\subset W$ and $B'|_{V}=B$. Then, $\CZ_{s}(\cl(V,B),\cl(W,B'))\cong \cl(V^{\perp},B^{\perp})$, where $V^{\perp}$ is the orthogonal complement to $W$ in the sense of the bilinear form $B'$, and $B^{\perp}=B'|_{V^{\perp}}$.
\end{proposition}
The proof follows from matching the conditions of \eqref{grcom} and the Clifford relations \eqref{ClifRels}.

Using this proposition, we immediately obtain the graded commutants corresponding to the types $\IR\IR$, $\IR\IC$, $\IR\IH$, $\IC\IR$, $\IC\IH$, $\IH\IR$ in Table \ref{table3}. The four remaining cases should be computed separately.

We begin with $\CZ_{s}(\cl_{-2}, \cl_{3,-1})$. Let $\cl_{3,-1}$ be generated by $x_1, x_2, x_3$, and $y$ such that $x_i^{2}=\unit$ and $y^2=-\unit$. Then, the superalgebra $\cl_{-2}$ generated by odd anticommuting elements $a$ and $b$ and squaring to $-\unit$, can be embedded into $\cl_{3,-1}$ in the following way:
\begin{equation}
    a\to y\;,\quad b \to x_1x_2x_3\;.
\end{equation}
One can verify that the graded commutant is purely even and contains $\unit, x_1x_2, x_1x_3, x_2x_3$, what implies $\CZ_{s}(\cl_{-2},\cl_{3,-1})\cong \IH$. We work out the case $\IH\IH$ next. Let $\cl_{4,-2}$ be generated by $x_i$ with $1\le i\le 4$ and $y_1, y_2$ such that $x_i^2=\unit$ and $y_{j}^2=-\unit$. We generate $\cl_{-3}$ with three odd elements $a,b,c$, and embed it into $\cl_{4,-2}$ as follows:
\begin{equation}
    a\to y_1\;,\quad b\to y_2\;,\quad c\to x_1x_2x_3\;.
\end{equation}
As a result, the commutant contains $x_4, ~x_1x_2x_4,~ x_1x_3x_4,~ x_2x_3x_4$, what implies\\ $\CZ_{s}(\cl_{-3},\cl_{4,-2})\cong \cl_{-3}$.

 We proceed with the case of a $\phi$-irrep of the type $\IC\IC_2$. Let $\cl_{3,-1}$ be generated by $x_1, x_2, x_3, y$ such that $x_i^2=\unit$ and $y^2=-\unit$. In its turn, $\IC\ell_1$ as a real superalgebra contains four elements: $\unit$, $i$, $e$, $ie$ such that $e$ and $ie$ are odd, commuting, and squaring to $\unit$ and $-\unit$ correspondingly. There is the embedding of real superalgebras:
\begin{align}
    \unit &\to \unit\;,\quad\quad \;\;\;e \to x_1\;,\\
    i&\to x_2x_3\;,\quad ie \to x_1x_2x_3\;.
\end{align}
One can verify that the supercommutant for this embedding  $\CZ_{s}(\IC\ell_1,\cl_{3,-1})$ contains $\unit$, $x_2x_3$, $y$, $x_2x_3y$, and can be identified with $\IC\ell_1$. The same computation can be performed for the type $\IC\IC_1$, where $\CA_{s}\sim \IC\ell _2$ and obtain $\CZ_{s}(\IC\ell_2, \cl_{3,-1})$. This finishes the calculations needed to obtain Table \ref{table3}. Clearly, we can use this technique in order to reproduce the similarity type of $\CA_{s}$ from $\CW$. This proves that $\gamma$ defined in \eqref{gamma} is an isomorphism of sets.

\begin{Remark}
The procedure of taking the graded commutant described above is nothing but solving the following equation in the graded Brauer monoid of reals (see Appendix \ref{Brauer}):
\begin{equation}
\label{eq.inBR}
    [\CA_{s}] [\CW] =[\cl_{2}]\;.
\end{equation}
Taking the graded commutant within $\End_{\IR}(V_{\IR},I)$ can be seen as a reflection in the graded Brauer monoid across the element 2.
\end{Remark}
 Our proof can be recast into the following diagram, which is true for any $\phi$-irrep $(\rho, V_{\IR},I)$:
\begin{center}
\begin{tikzpicture}[node distance=2.5cm, auto]
  \node (P) {$(\rho, V_{\IR},I)$};
  \node (B) [right of=P] {$\mbox{Dyson Type}$};
  \node (C) [below of=B] {$[\CA_{s}]$};
  \node (P1) [below of=P] {$[\CW]$};
  \draw[->] (P) to node {$\alpha$} (B);
  \draw[->] (B) to node {$\beta$} (C);
  \draw[->] (C) to node {$\gamma$} (P1);
  \draw[->] (P) to node {$\delta$} (P1);
\end{tikzpicture}
\end{center}
 where $\beta$ and $\gamma$ are   isomorphisms of sets and  the diagram commutes. We explained how to construct the maps $\alpha$, $\beta$, and $\gamma$, and showed that $\beta$ and $\gamma$ are bijective. The fact this is true for any $\phi$-irrep of any $\IZ_2$-graded group and makes the diagram above commutative.

\section{A Frobenius-Schur Indicator For \texorpdfstring{$\phi$}{TEXT}-irreps}
In the original paper \cite{Dyson_1962}, Dyson defined classical Frobenius-Schur indicators for the algebras $\CD$ and $\CA$. Recall, for any unitary representation $(\rho^{\IC},V)$ of the group $G^{0}$, we can define a Frobenius-Schur indicator
\begin{equation}
    \Pi(\rho^{\IC})=\frac{1}{|G^0|}\sum_{u\in G^{0}}\mbox{Tr}(\rho^{\IC}(u^2))\;.
\end{equation}
One can prove that the F-S indicator on any irreducible representation takes one of the three values: $\{-1,0,+1\}$. According to the value of the F-S indicator, we call an irrep potentially quaternionic, complex, or potentially real. As we reviewed above, the $\phi$-irrep, restricted to its unitary piece, might be reducible. However, it contains at most two components of the same type. Then, we can apply the Frobenius-Schur indicator to any of the two irreducible unitary subrepresentations with the same result. We always assume that the F-S indicator is applied to only one component. The three possible values of the F-S indicator correspond to one of the three possible similarity types of the simple component of $\CA$.

Further, Dyson (following V. Bargmann) introduced an indicator distinguishing types of $\CD$. It is defined for any $\phi$-irrep in the following way:
\begin{equation}
    \Pi'(\rho^{\IC})=\frac{1}{|G^1|}\sum_{a\in G^{1}}\mbox{Tr}(\rho^{\IC}(a^2))\;,
\end{equation}
It was proven by Dyson that this indicator also takes the three possible values $\{-1,0,+1\}$, corresponding to the three possible types of $\CD$: $\IH$, $\IC$, and $\IR$ respectively.

These two indicators can be assembled into one indicator
\begin{equation}
    \pi= \Pi'(\rho^{\IC})+i\Pi(\rho^{\IC})\;.
\end{equation}
which has remarkable properties. Set $\omega=0$ if $\pi =0$ and   $\omega\equiv \pi/|\pi|$
when $\pi$ is nonzero. Now, the
 graded Brauer group of reals $Br(\IR)$ is isomorphic to $\IZ_{8}$ which we identify
as the multiplicative group of eighth roots of unity. Choosing a suitable
isomorphism, the image of $\CW$ in $Br(\IR)$ will be identical to the value of $\omega$.
That is, $\omega$ takes the following values depending on the Dyson Type of the $\phi$-irrep:
\begin{align}
\label{super-F-S indicator}
    \omega=\begin{cases}0\;,\quad \quad \quad \quad \quad \quad \IC\IC_1,\; \IC\IC_2\;,\\
    \exp(2\pi i n/8 )\;,\quad \mbox{others}\;.
    \end{cases}
\end{align}
Put differently, we have the following table:
 \begin{table}[h!]
\begin{center}
\begin{adjustbox}{width=6.5cm}
\begin{tabular}{ |c|c|c|c|c|c|c|c| }
\hline
$\CW $ & $Type$  & $\Pi'$ & $\Pi$ & $\omega$ \\
  \hline
 $\cl_{1}$ & $\mathbb{RR}$ & $+1$ & $+1$ & $\exp(\scriptstyle{ +1\cdot \frac{2\pi i}{8}})$  \\
 \hline
$\IR$ & $\mathbb{RC}$  & $+1$ & $0$ & $\exp(\scriptstyle{ 0\cdot \frac{2\pi i}{8}})$ \\
\hline
$\cl_{-1}$ & $\mathbb{RH}$  & $+1$ & $-1$ & $\exp(\scriptstyle{-1\cdot \frac{2\pi i}{8}})$\\
\hline
$\cl_{2}$ & $\mathbb{CR}$  & $0$ & $+1$ & $\exp(\scriptstyle{+2\cdot\frac{2\pi i}{8}})$ \\
\hline
$\cl_{-2}$ & $\mathbb{CH}$  & $0$ & $-1$ & $\exp(\scriptstyle{-2\cdot \frac{2\pi i}{8}})$  \\
\hline
$\cl_{3}$ & $\mathbb{HR}$  & $-1$ & $+1$ & $\exp(\scriptstyle{+3\cdot  \frac{2\pi i}{8}})$  \\
\hline
$\IH$ & $\mathbb{HC}$  & $-1$ & $0$ & $\exp(\scriptstyle{+4\cdot  \frac{2\pi i}{8}})$   \\
\hline
$\cl_{-3}$ & $\mathbb{HH}$  & $-1$ & $-1$ & $\exp(\scriptstyle{-3\cdot  \frac{2\pi i}{8}})$  \\
\hline
\end{tabular}
\end{adjustbox}
\caption{The Dyson Types having real central graded commutants matched with their Frobenius-Schur indicators.}
\label{tableFS}
\end{center}
\end{table}

A similar indicator was introduced in \cite{Gow1979},
and rediscovered recently in a study of crossed product group superalgebras
by Ichikawa and Tachikawa in  \cite{ichikawa2020super}.

\section{How To Construct Examples of Different Dyson Types of \texorpdfstring{$\phi$}{phi}-irreps From Unitary Representations}
\label{HowToConstruct}
Here, we discuss one particular way of obtaining $\phi$-irreps of a given type from a complex representation.

The simplest example comes from a graded group being a direct product $\IZ_{2}\times G^{0}$. Let $(\rho^{\IC},V)$ be a complex unitary irrep of the group $G^{0}$. In the following five cases, we map the non-trivial element of $\IZ_2$ to an element of $\CX$ of the given unitary irrep, or that of its double-copy. We begin with the type $\IR\IR$. If $\rho^{\IC}$ is a potentially real irrep, then the commutant $\CX$ contains an antilinear operator squaring to $+\unit$. We simply include this operator to our irrep and obtain a $\phi$-irrep of the type $\IR\IR$. In order to obtain the type $\IH\IR$ from a potentially real irrep, we have to take its double-copy $(\rho^{\IC}\ot \unit_2, V\ot \IR^2)$. This way we enlarge the commutant and it now includes an antiunitary operator squaring to $-\unit$, which we use in order to get a $\phi$-irrep of the type $\IH\IR$.

Assume $\rho^{\IC}$ is a potentially quaternionic irrep. Then, we can add the antilinear element of $\CX$ squaring to $-\unit$ and obtain a $\phi$-irrep of the type $\IR\IH$. Analogously, we can take a double-copy, such that the commutant contains an element squaring to $+\unit$. Adjoining this gives us the type $\IH\IH$.

Finally, the commutant of a ``complexification'' of a truly complex irrep always contains an antilinear operator exchanging the two irreducible components, that allows us to obtain a $\phi$-irrep of the type $\IC\IC_2$. One can check that a complex structure still belongs to the commutant $\CZ$.

There is a universal way of constructing $\phi$-irreps from a complex-linear irrep $(\rho^{\IC},V)$ of a group $H$. We define $G^{0} \equiv H \times H$ and its representation $(\rho^{\IC}\oplus \rho^{\IC}, V\oplus V)$. There is an automorphism of $G^{0}$ permuting the factors of $H$. Let us denote the non-trivial generator of $\IZ_{2}$ by $p$, then
\begin{align}
    \IZ_{2}&\to Aut(G^{0})\\
    p: (h_1,h_2) &\mapsto (h_2,h_1)\;, \;\; \mbox{where}\;h_1,h_2\in H\;,
\end{align}
what allows us to define a semidirect product $G\equiv \IZ_2 \ltimes G^{0}$. The group $G$ contains pairs $[t,(h_1,h_2)]$ with $t \in \IZ_2$ and is clearly $\IZ_2$-graded. We define a $\phi$-irrep of $G$ by the following rule:
\begin{equation}
    \rho([p,(h_1,h_2)]=\rho^{\IC}(h_2)J\oplus \rho^{\IC}(h_1)J\;,
\end{equation}
where $J$ is a complex-conjugation operator acting on $V$. If $I$ is the complex structure for $V$, then the commutant $\CZ$ contains precisely one non-trivial operator $I\oplus (-I)$. On the other hand, the type of $\rho^{\IC}$ does not matter for this construction and we obtain any of the types $\IC\IR$, $\IC\IH$, or $\IC\IC_1$.

There are two remaining cases where $G$ should be taken as the general semidirect product group. Given a truly-complex irreducible representation of $G^{0}$, we can map the non-trivial generator of $\IZ_2$ to a real or quaternionic structure and obtain a $\phi$-irrep of the type $\IR\IC$ or $\IH\IC$ correspondingly.

\subsection{A Special Case: Groups Containing ``Time-Reversal''}
\label{Sec6}
In this section, we review Dyson's treatment of $\phi$-irreps of ``factorizable'' groups. It refers to groups of symmetries containing a distinguished time-reversal operation commuting with all unitary transformations. We note that our discussion of antiunitary symmetries is not limited to operations of time-reversal, and time-reversal symmetries are not limited to those commuting with unitary symmetries.

Until this point we made an assumption that the group $G$ is non-trivially graded. We assumed that map \ref{gradingG} is surjective, or equivalently, the following sequence is exact  \begin{equation*}
     \begin{tikzcd}
1 \arrow[r]
& G^{0} \arrow[r]  & G \arrow[r, "\phi"]  & \mathbb{Z}_2 \arrow[r]  & 1
\end{tikzcd}
 \end{equation*}
Let us consider a special case when $G$ is a semidirect products of $G^{0}$ with $\IZ_2$, so that the sequence is split. In other words, there exists a group homomorphism
\begin{equation}
\label{splittinghom}
   s: \IZ_2 \to G\;
\end{equation}
such that $\phi \circ s=id_{\IZ_2}$. Then, $G$ is a split extension; they are known to be classified by group homomorphisms
\begin{equation}
\label{semidirect}
   \psi: \IZ_2 \to \mbox{Aut}(G^0)\;,
\end{equation}
allowing as to define the semidirect product $G\cong  G^0\rtimes_{\psi} \IZ_2$.

\begin{Remark}
The first example of a group admitting $\IZ_2$-grading and not being a semidirect product is the quaternionic group $\CQ_8$. The normal subgroup $G^{0}$ here is generated by $\{\pm 1, \pm i\}\cong \IZ_4$. Accordingly, the only order-2 subgroup of $\CQ_8$ is the one generated by $\{\pm 1\}$. Those two subgroups intersect non-trivially and cannot form a semidirect product. Then, $\CQ_8$ is not isomorphic to $\IZ_4 \rtimes \IZ_2$.
\end{Remark}

Furthermore, we make an assumption that every automorphism of $G^{0}$ is inner, or equivalently, the group of outer automorphisms is trivial. Then, there always exists $\tilde {g}\in G^{0}$ such that
\begin{equation}
\label{psi1}
    \psi(P)(g)=\tilde {g}^{-1} g \tilde {g}\;,\quad \mbox{for} \;g\in G^{0}\;.
\end{equation}
and $P$ is the non-trivial element of $\IZ_2$. On the other hand, we can use the splitting homomorphisms in order to construct the map \eqref{semidirect}
\begin{equation}
\label{psi2}
    \psi(P)(g)=s(P)^{-1}g\;s (P)\;.
\end{equation}
 For any semidirect product $G\cong G^0\rtimes_{\psi} \IZ_2$, with a group $G^0$ having only inner automorphisms, specified by the homomorphism $s$ and $\tilde{g}$, we can define an operator $T=s(P)\tilde {g}^{-1} \in G^{1}$. Matching \ref{psi1} and \ref{psi2}, we see that $T$ commutes with $G^{0}$.

\begin{Remark}
The smallest example of a group which is a semidirect product and not every automorphism is inner is the dihedral group of 8 elements $D_{4}$. Here, $G^{0}=\IZ_4$ and it has one non-trivial outer automorphism, the one used to build $D_{4}$. Seen geometrically, this is the group of symmetries of the square. If we represent the elements involving reflections antiunitary, there will be no distinguished ``time-reversal''.
\end{Remark}
Let us note that such $T$ is well-defined in the case when both $G^{0}$ and $\IZ_2$ are normal subgroups of $G$. In this case, $G\cong G^{0}\times \IZ_2$ and $T$ is simply the non-trivial generator of $\IZ_2$. Then, we impose no restrictions on $G^{0}$. Another situation, if $G^{0}$ has trivial group of outer automorphisms and trivial center ($G^{0}$ is complete), $G$ is automatically a direct product of $G^{0}$ and $\IZ_2$.

Let us call groups satisfying any of two properties above factorizable. For a factorizable symmetry group, there always exists an operator $T$ commuting with all unitary symmetries. This operator is canonically associated with the operation of time-reversal in physics. If the symmetry group $G$ is a semidirect product and not factorizable, one may voluntarily call any antiunitary symmetry as the time-reversal, but there is no canonical choice.

Let us consider irreducible $\phi$-representation of a factorizable group. It immediately follows for such $\phi$-irrep that $\CX$ is not a one-dimensional algebra. Indeed, the element $\rho(T)$ commutes with all $\CA$. However, this element does not belong to $\CY$ as it is $\IC$-antilinear. So, a $\phi$-irrep of a factorizable group cannot have a purely even graded commutant.

We proceed further turning to group superalgebras. Here, $\rho(T)$ can be associated with the element $u$ such that $\CA_{s}^{1}=u \CA^{0}$ and $u\in \CX$. Matching with Table \ref{table2} and Table \ref{table3}, we imply that $\phi$-irreps of factorizable groups are always of the following Dyson Types: $\IR\IR$, $\IR\IH$, $\IC\IC_2$, $\IH\IR$, $\IH\IH$. In conclusion, the group superalgebra of graded factorizable groups must be similar to one of the following: $\cl_{\pm 1}, \cl_{\pm 3}, \cl_{1,-1}, \cl_4$.

Finally, we can discuss the sign of $\rho(T)^2$ when $A_{s}^{0}$ is a central simple algebra over $\IR$. According to that, $\rho(T)^2=+\unit$ for the Dyson Types $\IR\IR$ and $\IH\IH$, and $\rho(T)^2=-\unit$ for the types $\IR\IH$ and $\IH\IR$. This way, the sign of $\rho(T)^2$ is fixed by the full Dyson Type of the $\phi$-irrep and not solely by the type of $\CD$ or $\CA$.
\appendix

\section{General Facts About Superalgebras}
Here, we give basic definitions concerning superalgebras. Taking even components everywhere, would give us definitions concerning ordinary algebras. Vectors spaces below are defined over a field $\kappa$, for our needs it will stand for $\IR$ or $\IC$. Throughout the paper, all superalgebras are assumed to be unital and associative.
\label{B}
\begin{definition}
a) A $\kappa$-superalgebra $A$ is a vector superspace over $\kappa$ together with a product $a\otimes a'\to aa'$. For homogeneous elements of $A$,
\begin{equation}
     \mbox{deg}(aa')=\mbox{deg}(a)+\mbox{deg}(a')\;.
\end{equation}

b) Two homogeneous elements of a superalgebra are said to graded-commute, or super-commute provided
\[aa'=(-1)^{\scriptstyle{deg}(a)\scriptstyle{deg}(a')}a'a\]
If every pair of elements $a$ and $a'$ in $A$  graded-commute, then the superalgebra $A$ is called  supercommutative.

\begin{definition}
Algebra $A^{\scriptstyle{opp}}$ is said to be opposite to $A$ if it consist of the same elements, but multiplication is performed in opposite order.
\end{definition}

\begin{definition}
The center $\CZ(A)$ of a superalgebra $A$ is the maximal commutative ungraded subalgebra of $A$:
\begin{equation}
    z\in \CZ(A) \Leftrightarrow z a =a z \quad \mbox{for any $a\in A$}\;.
\end{equation}

The supercenter $\CZ_{s}(A)$ of a superalgebra $A$ is the maximal supercommutative graded subalgebra of $A$. Homogeneous elements of the supercenter $z^{0}, z^{1} \in A$ satisfy:
\begin{align}
    z^{0}\in \CZ_{s}^{0}(A) \Leftrightarrow z^{0}a^{0}=a^{0}z^{0}\;,\quad \& \quad  z^{0}a^{1}=+a^{1}z^{0}\;,\\
    z^{1}\in \CZ_{s}^{1}(A) \Leftrightarrow z^{1}a^{0}=a^{0}z^{1}\;,\quad \& \quad  z^{1}a^{1}=-a^{1}z^{1}\;.
\end{align}

A superalgebra is called supercentral if $\CZ_{s}(A)$ coincides with $\kappa$.
\end{definition}
\end{definition}
\begin{definition}
Let $A$ and $B$ be superalgebras. We define the graded tensor product $A\hat \ot B$ as the superalgebra, which is the graded tensor product as a vector space and the multiplication of homogeneous elements satisfy
\begin{equation}
   ( a_1\hat \ot b_1)(a_2 \hat\ot b_2)=(-1)^{\scriptstyle{deg}(a_2)\scriptstyle{deg}(b_1)}(a_1a_2)\hat\ot (b_1b_2)\;.
\end{equation}
\end{definition}
\begin{definition}
An ideal $I$ of superalgebra $A$ is called homogenious if it has the form $I=I^0\oplus I^1$ such that $I^0\subset  A^0$ and $I^1\subset A^1$. A superalgebra $A$ is called simple if its only homogeneous two-sided ideals are ${0}$ and $A$ itself.
\end{definition}
Simple superalgebras are objects of our main interest. In their turn, Clifford algebras are main examples of simple superalgebras.

\section{Clifford Algebras}
\label{Clifford algebras}
\begin{definition}
Let $V$ be a vector space over a field $\kappa$ and let $B$ be a non-degenerate symmetric bilinear form on $V$ valued in $\kappa$. The Clifford algebra $\cl(V_{\kappa},B)$ is the algebra over field $\kappa$ generated by the basis ${e_i}$ for $V_{\kappa}$ and defined by relations
\begin{equation}
\label{ClifRels}
    e_ie_j+e_je_i=2 B(e_i,e_j)\;.
\end{equation}
\end{definition}
 Even subalgebra of $\cl(V_{\kappa},B)$ consist of products of even numbers of generators, including the base field, while odd part consists of products of odd number of generators. The grading is consistent with multiplication because the Clifford structure is given by the quadratic relations. For Clifford algebras, it is easy to find the opposites using the sign rule (remember that the multiplication is a bilinear map $A\hat \ot_{\kappa} A\to A$).
\begin{equation}
\cl(V_{\kappa}, B)^{\scriptstyle{opp}}\cong \cl (V_{\kappa},-B)\;.
\end{equation}
The most interesting examples for us are real and complex Clifford algebras. Assume we fixed the basis of $V_{\kappa}$ such that $B$ is diagonal. Then, over the reals this form has the signature $(r,s)$, where $r$ is the number of positive elements on the diagonal and $s$ is the number of negative ones. We denote the real Clifford algebra corresponding to the bilinear form of signature $(r,s)$ by $\cl_{r,-s}$. If the form is positive-definite (negative-definite), then we use the following one-index notation:  $\cl_{r,0}\equiv \cl_{r}$ and $\cl_{0,-s}\equiv \cl_{-s}$.

Without loss of generality, we can diagonalize the form $B$ over the complex numbers in such a way that it consists of $+1$'s on the diagonal. We denote the complex Clifford algebra corresponding to a rank-$n$ bilinear form by $\IC\ell _n$.

As follows from the Clifford relations \eqref{ClifRels}, the graded tensor product of two Clifford algebras is a Clifford algebra:
\begin{equation}
    \cl(V_{\kappa}', B')\hat \ot_{\kappa} \cl(V_{\kappa}'', B'')\cong \cl (V_{\kappa}'\oplus V_{\kappa}'',B'\oplus B'')\;.
\end{equation}

\section{(Graded)-Brauer Groups}
\label{Brauer}

 A theorem by Wedderburn provide us with a structure of simple (super)-algebras. In order to formulate that theorem, we need the notion of a division superalgebra.
\begin{definition}
A division superalgebra $A$ is a superalgebra such that any non-zero homogeneous element is invertible.
\end{definition}
\begin{theorem}(super-Wedderburn)
\label{super-Wed}
Any simple superalgebra over the field $\kappa$ is a matrix superalgebra with coefficients in a division superalgebra over the field $\kappa$.
\end{theorem}
A very good source to look up the proof is the book by V.S. Varadarajan \cite{varadarajansupersymmetry}. For our applications, we invoke the theorems by Frobenius and Wall, listing all real division algebras and superalgebras respectively.
\begin{theorem}{(Frobenius)}
\label{FrobTh}
There are three associative real division algebras: $\mathbb{R}, \mathbb{C}$, and $ \mathbb{H}$.
\end{theorem}
\begin{theorem}{(Wall, Deligne)}
\label{Wall}
There are ten associative real division superalgebras: three purely even superalgebras $\mathbb{R}, \mathbb{C}$, and $ \mathbb{H}$ and seven Clifford superalgebras $\mathbb{C}\ell_{1}, C\ell_{\pm1}, C\ell_{\pm2}, C\ell_{\pm3}$.
\end{theorem}
We sketch a proof of the Deligne-Wall theorem below. A consequence of those theorems is that any real simple algebra has the form $\mbox{End}_{\IR}(V_{\IR})\ot_{\IR} \mathbb{K}$, where $\mathbb{K}$ is one of the three algebras: $\IR$, $\IC$, and $\IH$. Any real simple superalgebra is isomorphic to $\mbox{End}_{\IR}(V_{\IR}^{m|n})\hat \ot_{\IR} \mathbb{D}$, where $\mathbb{D}$ is one of the ten real division superalgebras above.

The following relation is an important consequence of the Wedderburn theorem (see e.g. \cite{varadarajansupersymmetry}):
\begin{equation}
\label{Wedderburn2}
\mbox{End}_{\IR}(V_{\IR}^{m|n})\;\hat \otimes\;\mathbb{D}\; \hat \otimes \mbox{End}_{\IR}(V_{\IR}^{p|q})\;\hat \otimes\;\mathbb{D}^{\scriptstyle{opp}}\cong \mbox{End}_{\IR}(V_{\IR}^{mp+nq|mq+np})\hat \otimes \End_{\IR}(\mathbb{D})\;,
\end{equation}
where $\End_{\IR}(\mathbb{D})$ is the endomorphism superalgebra of $\mathbb{D}$ considered as a vector superspace.

The (sueper)-Wedderburn theorem provides us with the general structure of simple (super)-algebras over a field $\kappa$. It is suggestive to introduce an equivalence relation on the superalgebras, such that $A$ and $B$ are similar iff they are isomorphic to matrix superalgebras with coefficients in the same division $\kappa$-superalgebra $\mathbb{D}$:
\begin{equation}
    A\sim B \iff A\cong \End_{\kappa}(V_{\kappa}^{m|n})\hat \ot_{\kappa}\mathbb{D} \quad \& \quad B\cong\End_{\kappa}(V_{\kappa}^{p|q})\hat \ot_{\kappa}\mathbb{D}\;.
\end{equation}
We denote the image of $A$ in the set of similarity classes by $[A]$.

\begin{Remark}
Throughout the paper, we interchangeably use the terms similarity classes and Morita-classes. There is a general notion of Morita equivalence of algebras, such that two (graded) algebras are Morita-equivalents iff their categories of (graded) modules are equivalent. It turns out that the equivalence relation we introduced is identical to the Morita-equivalence relation.
\end{Remark}

It is clear from the definition that the set of similarity classes of simple (super)-algebras is in one-to-one correspondence with the set of division (super)-algebras over $\kappa$. So, there are 3 classes of simple real algebras and 10 classes of simple real superalgebras. As we will show later, there is a unique class of simple complex algebras and two classes of simple complex superalgebras. It is worth noting that the set of similarity classes of simple algebras over $\kappa$ contain similarity classes of algebras over all algebraic extensions of $\kappa$ up to its algebraic closure.

 A natural operation on simple algebras is the tensor product. Over the real numbers, we can write down the following tensor products
\begin{equation}
    \IR\ot_{\IR} \IR \cong \IR, \quad \IH\ot_{\IR} \IR \cong \IH\;,\quad \IH\ot_{\IR}\IH\cong \End_{\IR}(\IR^{4})\;,
\end{equation}
while $\IC\ot_{\IR}\IC\cong \IC\oplus \IC$. The subtlety here is that $\IC$ is not a central simple algebra over $\IR$. It is an important fact, which we do not prove here, that a (graded) tensor product of two simple (super)-central (super)-algebras is a simple (super)-central (super)-algebras. Moreover, Morita-classes of (super)-central simple (super)-algebras (CSA, SCSSA) form a group. As we mentioned, CSA (SCSSA) are closed under (graded) tensor product, and there always exists an opposite due to the relation \eqref{Wedderburn2}.

Finally, we can list the four relevant (graded) Brauer groups:
\begin{equation}
    Br(\IR)\cong \IZ_2\;,\quad Br(\IC)\cong \IZ_1\;,\quad sBr(\IR)\cong \IZ_8\;,\quad sBr(\IC)\cong \IZ_2\;.
\end{equation}
Further, we notice that a tensor product over $\IR$ of a real CSA with a complex CSA gives us a complex CSA. Then, a pair of the Brauer groups above form a monoid, see the blog by John Baez \cite{Baez2014}. We call it a Brauer monoid $BR(\IR)$: as a set, it is a disjoint union $BR(\IR)\cong Br(\IR) \coprod Br(\IC)$. The Brauer monoid of reals is isomorphic to the multiplicative monoid of three elements:
\begin{equation}
    BR(\IR)\cong \{-1,0,1\}\;,
\end{equation}
where we identify $[\IH]\to -1$, $[\IC]\to 0$, $[\IR]\to +1$. Analogously, we can form a monoid $sBR(\IR)$ of real and complex graded Brauer groups. $sBR(\IR)$ is isomorphic to the additive monoid on ten elements $\{0,1,2,3,4,5,6,7,8, \textbf{0}, \textbf{1}\}$, where we identify
\begin{equation}
    [\cl_{r,-s}]\to r-|s|\; \mbox{mod} \;8\;,
\end{equation}
\begin{equation}
    [\IC\ell_{n}]\to \textbf{n}\; \mbox{mod} \;2\;.
\end{equation}
Addition of lightface numbers is the same as in $\mathbb{Z}_8$, and addition of boldface numbers is the same as in $\mathbb{Z}_2$. In order to add a lightface number $x$ and a boldface number $\textbf{y}$, we should firstly project $x$ into $\mathbb{Z}_2$. For example, in this monoid  $2+\textbf{1}=\textbf{1}$.

\subsection{Wall's theorem}\label{RDS}
In this section, we prove Walls' theorem \cite{Wall1964}, following a more up-to-date treatment by Deligne \cite{Deligne_noteson}.

Let us first notice that the only division algebra $A_{\bar \kappa}$ over an algebraically closed field $\bar \kappa$ is $\bar \kappa$ itself. Any division algebra should be finite-dimensional, so let us consider a minimal polynomial containing powers of $x\in A_{\bar \kappa}$.
\begin{equation*}
    f(x)=x^{n}+ a_1 x^{n-1}+...a_{n}=0\;.
\end{equation*}
Since $\bar \kappa$ is algebraically closed, $f(x)$ has at least one root $b$ and can be factorized $f(x)=(x-b)g(x)=0$. Since $f(x)$ was minimal, $g(x)\neq 0$ and $x=b$, what implies $x\in \bar \kappa$. So, the only division $\bar \kappa$-algebra is $\bar \kappa$ itself.

Now, let us classify all division superalgebras over $\bar \kappa$. We recall that the definition of a division superalgebra requires homogeneous non-zero elements to be invertible. We can apply the same logic as in the previous paragraph to even and odd components of a division $\bar \kappa$-superalgebra and obtain the unique algebra $\bar \kappa [\epsilon]$, where $\epsilon$ is odd and squares to 1. In particular, we obtain all complex division superalgebras: $\IC$ and $\IC \ell_{1}$.

\begin{theorem}(Wall, Deligne)
\label{Wall-Deligne}
The graded Brauer group $sBr(\IR)$ is an iterated extension of $\mathbb{Z}_2$ by $\IR^*/\IR^{*2}$ by the ordinary group $Br(\IR)$.
\end{theorem}
More explicitly, the theorem says that $sBr(\IR)$ fits into the exact sequence:
 \begin{equation}
     \begin{tikzcd}
1 \arrow[r]
& sBr(\IR)' \arrow[r,"\phi_1"]  & sBr(\IR) \arrow[r, "\phi_2"] & \mathbb{Z}_2 \arrow[r]  & 1
\end{tikzcd}
 \end{equation}
 while $sBr(\IR)'$ itself fits into the following exact sequence:
\begin{equation}
     \begin{tikzcd}
1 \arrow[r]
& Br(\IR) \arrow[r,"\phi_3"]  & sBr(\IR)' \arrow[r,"\phi_4"] & \IZ_2 \arrow[r]  & 1
\end{tikzcd}
 \end{equation}
\begin{proof}
We first show the existence of the group homomorphism $\phi_2$. This map is just an extension of scalars, it maps all CS superalgebras belonging to a given similarity class to their complexifications. It induces the group homomorphism:  $sBr(\IR)\to sBr(\IC)\cong \mathbb{Z}_2$, which is surjective as $\IC[\epsilon]$ is the image of $\IR[\epsilon]$.

To make the sequence exact, the first term $sBr(k)'$ has to be isomorphic to the kernel of $\phi_2$. These are similarity classes of SCSSA such that, after extension of scalars,
\begin{equation}
    A\to A\ot{\IC}\equiv A_{\IC}
\end{equation}become isomorphic to $\End(\IC^{m|n})$ .

Now, let us analyze $sBr(k)'$. A matrix superalgebra $\End(\IC^{m|n})$ has precisely two simple graded modules: $\IC^{m|n}$ and $\Pi \IC^{m|n}\cong \IC^{n|m}$. Let us denote the set of simple modules over $A_{\IC}$ by $I(A)$. We can notice that simple modules over $A_{\IC}$ can be seen as simple $A$-modules. Also, $I(A)$ admits an action of the Galois group $Gal(\IC/\IR)\cong \IZ_2$. Reversing the complex structure does not necessarily lead to the same module, so that the Galois action might exchange the elements of $I(A)$. Thus, we have a map
\begin{equation*}
    \alpha_{A}: Gal(\IC/\IR)\to \mathbb{Z}_2\;,
\end{equation*}
where $\IZ_2$ is the group of permutations of the two elements of $I(A)$.

 If $[A_1]\in sBr(\IR)'$ and $[A_2]\in sBr(\IR)'$, then $[A_1\hat \ot A_2]\in sBr(\IR)'$.
Also, one can check that $\alpha_{A' \ot A''}=\alpha_{A'}+\alpha_{A''}$, so that the map
\begin{equation}
    A\to \alpha_{A}
\end{equation}
descends to
\begin{equation}
    sBr(\IR)'\to \mbox{Hom}(Gal(\IC /\IR), \mathbb{Z}_2)\cong \IZ_2\;.
\end{equation}
In order to show that this is a surjective homomorphism, we provide an example of a  non-trivial element of $\mbox{Hom}(Gal(\IC /\IR), \mathbb{Z}_2)$. Assume $A=\cl_2$, its complexification $\IC\ell_2\cong \End(\IC^{1|1})$ has the following graded representations (we pick a homogeneous basis):
\begin{equation}
    \rho_{1}^{\IC}(e_1)=
   \begin{pmatrix}
    0&1\\
    1&0\\
     \end{pmatrix}\;, \quad \rho_1^{\IC}(e_2)=
     \begin{pmatrix}
    0&-i\\
    i&0\\
     \end{pmatrix}
\end{equation}
\begin{equation}
    \rho_{2}^{\IC}(e_1)=
   \begin{pmatrix}
    0&1\\
    1&0\\
     \end{pmatrix}\;, \quad \rho_2^{\IC}(e_2)=-
     \begin{pmatrix}
    0&-i\\
    i&0\\
     \end{pmatrix}
\end{equation}
One can check that those two super-reps are non-equivalent. There is a standard argument: there exists an element $e_1e_2\in \CZ(\IC\ell _{2}^{0})$, which acts as a multiplication by $\pm i$ on the even subrepresentations. Reversing the complex structure reverses the sign of $\rho_i(e_1)\rho_i(e_2)$, what means that the Galois action permutes the irreps inside of the super-irrep. However, permutation of even subspaces is not an isomorphism of super-representations, so that complex conjugated representations of $\IC\ell_2$ are not equivalent.

The argument with a center holds for a general $A$ such that $[A]\in sBr(\IR)'$. It is always true that $\CZ(A^{0})$ acts on even subspaces of elements from $I(A)$ by multiplication by characters $(\chi_1,\chi_2)$ (where by a character of the module $M$ we mean a module homomorphism $\mbox{Hom}(M,\IC)$). The Galois action is trivial iff one of the characters is zero or they are both real numbers. The first case corresponds to a purely even $A$, while the second case correspond to real CS superalgebras having center isomorphic to $\IR\oplus \IR$. Indeed, complexification here is a simple multiplication by the complex unit and complex conjugation does not permute the characters.

Let us take any $[A]\in sBr(\IR)'$. The class of $A$ in $sBr(\IR)'$ is that of some real division superalgebra $\mathbb{D}$. If $\CZ(A^{0})\cong \IR \oplus \IR$, then $\CZ( \mathbb{D}^{0})\cong \IR \oplus \IR$, which is impossible as $\mathbb{D}^{0}$ is a real division algebra. So, the kernel of $\phi_4$ contains algebras similar to some real division algebra, which is nothing but the ordinary Brauer group of reals.
\end{proof}
Using theorem \ref{Wall-Deligne}, we derive $|sBr(\IR)|=8$. The simplest representative is $\cl_1$, which is clearly a real division superalgebra. In fact, tensor powers of $\cl_1$ exhaust all similarity classes of central simple superalgebras. The proof of this statement is very clearly presented in \cite{varadarajansupersymmetry}. We also recommend a paper by Todd Trimble \cite{Trimble2005}, where the theorem of Wall is proven by construction.

\end{document}